\documentclass[12pt, letterpaper]{article}

\usepackage[english]{babel}
\usepackage[utf8x]{inputenc}
\usepackage{authblk}
\usepackage{amsmath}
\usepackage{amssymb}
\usepackage{amsthm}
\usepackage{theoremref}
\usepackage{xcolor}
\usepackage{eucal}
\usepackage{flushend}

\newtheorem{definition}{Definition}

\newtheorem{lemma}{Lemma}
\newtheorem{theorem}{Theorem}
\newtheorem{remark}{Remark}

\newcommand{\Fqm}{\mathbb{F}_{q^m}}

\newcommand{\Fq}{\mathbb{F}_{q}}

\newcommand{\bs}{\mathbf{s}}
\newcommand{\be}{\mathbf{e}}
\newcommand{\bx}{\mathbf{x}}
\newcommand{\by}{\mathbf{y}}

\newcommand{\bh}{\mathbf{h}}
\newcommand{\bu}{\mathbf{u}}
\newcommand{\bv}{\mathbf{v}}

\newcommand{\calC}{\mathcal{C}}
\newcommand{\calE}{\mathcal{E}}
\newcommand{\calH}{\mathcal{H}}
\newcommand{\calA}{\mathcal{A}}
\newcommand{\calB}{\mathcal{B}}
\newcommand{\calV}{\mathcal{V}}
\newcommand{\calU}{\mathcal{U}}
\newcommand{\calT}{\mathcal{T}}

\begin{document}

\title{Low Row Rank Parity Check Codes}

\author[1]{Ermes Franch}
\author[1]{Chunlei Li}
\affil[1]{Department of Informatics, University of Bergen, Norway}
\affil[]{email: \textit{\{ermes.franch,chunlei.li\}@uib.no}}

\date{January 2021}

\maketitle

\begin{abstract}
    In this paper we present an extended variant of low rank parity check matrix (LRPC) codes that have received significant interests in recent years. It is shown that the extension indeed yields a super-family of LRPC codes, which are termed low row rank parity check codes. The decoding method of the proposed codes is also investigated.
\end{abstract}

\section{Introduction}\label{sec:Introduction}
	Rank metric codes were independently introduced by Delsarte \cite{Delsarte:1978aa}, Gabidulin \cite{Gabidulin1985} and Roth \cite{roth1991maximum}.
	The Gabidulin codes, which are an equivalent of the Reed-Solomon codes for the rank metric, have found applications in network coding \cite{SilvaKschischangKoetter}, space-time coding \cite{Gabidulin-SpaceTimeCodes} and cryptography \cite{GPT-1991}. In 2013 Gaborit, Murat and Z\'{e}mor introduced a new family of rank metric codes, termed \textit{low rank parity check} (LRPC) codes, which can be seen as the equivalent of LDPC
	codes \cite{gallager1962}.
    With an efficient probabilistic decoding algorithm,
	LRPC codes have been applied in several cryptosystems \cite{IBE-rank2017,LRPC-2019-TIT, ROLLO} as well as a signature scheme \cite{RankSign2014}, which are designed
	for better resistance against quantum computers and smaller key size compared to the classic McEliece cryptosystem in Hamming metric \cite{McEleice1978}. Among these schemes, the cryptosystem ROLLO \cite{ROLLO}, based on a subfamily of LRPC codes, has reached the second round
    of the National
Institute of Standards and Technology (NIST) post-quantum competition.
	
\smallskip	
	
	In this paper we will present an extended variant of the LRPC codes, which are called \textit{low row rank parity check codes}, by applying the constraint of low rank weight only 
	on each row of the parity check matrix instead of on the whole parity check matrix as for LRPC codes. The proposed low rank parity check codes are shown to contain LRPC codes as a strictly proper subclass. Due to the loosened constraint, 
    the decoding technique of LRPC codes no longer works for the low row rank parity check codes.
    We investigate the decoding process of these new codes
	and discuss potential reasons for decoding failures.

\section{Background on rank metric codes}
\subsection{Notations}
Let $\Fq$ the finite field of $q$ elements, where $q$ is a power of a prime $p$. 
The extension field $\Fqm$ can be deemed as an $\Fq$-vector subspace of dimension $m$.
Given elements $v_1, \ldots, v_r$ in $\Fqm$,  
the $\Fq$-vector subspace generated by $v_i$'s, denoted in the calligraphic upper case, is given by
\[
\calV = \langle v_1, \ldots, v_r \rangle = \left\{\sum_{i=1}^r \lambda_i v_i \, | \, \lambda_i \in \Fq \right\}.
\]
We will denote vectors over  $\Fqm$ in bold lower case as $\mathbf{v}=(v_1, \ldots, v_n)$ and denote matrices in upper case as $M=(M_{i,j})$.
Given a $k \times n$ matrix $M$ and two non-empty ordered subsets $I \subseteq \{1,\ldots,k\}, J \subseteq \{1,\ldots,n\}$, the sub-matrix of $M$ whose rows are indexed by $I$ and columns are indexed by $J$ is denoted as $M_{I,J}$. The transpose of the matrix $M$ will be denoted as $M^\intercal$.

\smallskip

Let $\calV = \langle v_1, \ldots, v_d\rangle$ be an $\Fq$-vector space of dimension $d$. We denote by $\calV^*$ the set $\calV \setminus \{0\}$ and by $\calV^{-1}$ the set $ \{v^{-1}| v \in \calV^*\}$. Given another $\Fq$-vector space $\calU = \langle u_1, \ldots, u_d\rangle$ of dimension $r$, the product vector space of $\calV$ and $\calU$ is defined as
\(
\calV.\calU = \langle v_1u_1, \ldots, v_1u_d, \ldots, v_ru_1, \ldots, v_ru_d \rangle,
\) which has dimension at most $rd$.

\subsection{Basic definitions}
We start with the definitions of rank weight and rank distance of elements in $\Fqm$.
\begin{definition}
Let $S \subseteq \Fqm$ and let $\langle S \rangle$ be the $\Fq$-vector subspace generated by all the elements of $S$. The  \textbf{rank weight} of $S$ is given by  $\mathrm{rw}(S) = \dim(\langle S \rangle)$.
\end{definition}
We can extend this definition in a natural way to vectors and matrices.
\begin{definition}
Let $\bv = (v_1, \ldots, v_n)$ be a vector  over $\Fqm$. 
The $\Fq$-vector subspace $\calV = \langle v_1, \ldots, v_n \rangle$ is called  the \textbf{support} of $\bv$.
The \textbf{rank weight} of $\bv$ is given by 
\[ 
\mathrm{rw}(\bv) := \dim\left(\calV \right) = \dim(\langle v_1, \ldots, v_n \rangle).
\]
The \textbf{rank distance} between two vectors $\bu,\bv$ is given by
\[
\mathrm{d}(\bu,\bv) := \mathrm{rw}(\bu-\bv).
\]
\end{definition}
\begin{definition}\label{def:LRPC}
The \textbf{weight} $\delta$ of a matrix  $H$ over $\Fqm$ is defined as the rank weight of the set of all its elements.
The matrix $H$ is said to have \textbf{low weight} if $\delta < m$, particularly $\delta\ll m$.
A rank metric code $\calC$ over $\Fqm$ is said to be an \textbf{LRPC code} of weight $\delta$ if $\calC$ admits a parity check matrix $H$ of weight less than or equal to $\delta$.
\end{definition}
Given an LRPC code $\calC$ with a low-weight parity check matrix $H$, it is easy to build another  parity check matrix $H'$ for $\calC$ of full weight $m$. On the contrary, starting from a matrix $H'$ of weight $m$, it is hard to find an equivalent matrix $H$ of low weight. It would be at least as hard as finding $n-k$ linearly independent vectors of low weight in the dual code of $\calC$ (which are not guaranteed to exist if $\calC$ is not an LRPC code).


\section{The row-LRPC codes}
In this section we will introduce a super-family of LRPC codes by 
applying the low-weight constraint only to rows of a parity check matrix instead of the whole matrix. The decoding process of those new codes will be examined in the next section.
\begin{definition}\thlabel{def:row-LRPC}
The \textbf{row weight} $\rho$ of an $(n-k)\times n$ matrix $H$ over $\Fqm$ is defined as the maximum rank weight of all its rows:
\[
\rho = \max\limits_{1\leq i \leq n-k}(\mathrm{rw}(\bh_i)), 
\] where $\bh_i$ for $1\leq i \leq n-k$ is the $i$-th row of $H$.
The matrix $H$ is said to  have \textbf{low row weight} if $\rho<m$, particularly $\rho \ll m$.
We call a rank metric code $\calC$ a low row rank parity check (shortly row-LRPC)  code of row weight $\rho$ if $\calC$ admits a parity check matrix $H$ of \textbf{row weight} less than or equal to $\rho$.
\end{definition}
In the above definition we use the abbreviation row-LRPC codes, instead of LRRPC codes, for better distinguishing them from LRPC codes.
For simplicity of presentation, we denote by $\mathrm{LRPC}_q(n,m,k,\delta)$ the set of all the $[n, k]$ linear LRPC codes over $\Fqm$ of weight $\delta$, and by row-$\mathrm{LRPC}_q(n,m,k,\rho)$ the set of all the $[n, k]$ row-LRPC codes over $\Fqm$ of row weight $\rho$.

\begin{remark}
From Definition \ref{def:LRPC} and Definition \ref{def:row-LRPC}, we have $\mathrm{LRPC}_q(n,m,k,\delta) \subseteq \mathrm{LRPC}_q(n,m,k,\delta')$ if $\delta \le \delta'$, similarly we have that $\text{row-}\mathrm{LRPC}_q(n,m,k,\rho) \subseteq \text{row-}\mathrm{LRPC}_q(n,m,k,\rho')$ if $\rho \le \rho'$.
\end{remark}

Let $H$ be an $(n-k)\times n$ matrix with row weight $\rho$.  The weight $\delta$ of $H$ satisfies $\rho \leq \delta \leq (n-k)\rho$. Hence, a row-LRPC code of row weight $\rho$ may admit a parity check matrix with weight ranging from $\rho$ up to $m$ when $(n-k)\rho \geq m$.
In other words, it is clear that $\mathrm{LRPC}_q(n,m,k,\rho) \subseteq \text{row-}\mathrm{LRPC}_q(n,m,k,\rho).$
The opposite inclusion is not true. As shown below, under certain condition it is always possible to construct a row-LRPC code of row weight $\rho \geq 2$ that has all its parity check matrices of weight strictly greater than $\rho$.
\begin{theorem}
Let $2 \leq \rho < m$ and $2\rho  -2 \le k \le n-2$. Then we have $$\mathrm{LRPC}_q(n,m,k,\rho) \subsetneq \text{row-} \mathrm{LRPC}_q(n,m,k,\rho)$$
\end{theorem}
\begin{proof}
Suppose $\calA = \langle a_1, \ldots,a_\rho \rangle, \calB = \langle b_1, \ldots, b_\rho \rangle$ are two $\rho$-dimensional $\Fq$-vector subspace of $\Fqm$ such that there exists no $\lambda \in \Fqm$ satisfying $\calA = \lambda \calB $.
We build an $(n-k)\times n$ parity check matrix $H$ of row density $\rho$ in the following way
\[
H:=
\begin{pmatrix}
a_1 & \cdots & a_\rho & 0     & \cdots & 0         & \mathbf{0}_{n-2\rho}  \\
0   & \cdots & 0        & b_1   & \cdots & b_\rho  & \mathbf{0}_{n-2\rho} \\
0   & \cdots & 0        & 0     & \cdots & 0         & J  \\
\end{pmatrix},
\]
where $J$ is an ${(n-k-2)\times (n-2\rho)}$ matrix obtained by extending the identity matrix of order $n-k-2$ with $k+2-2\rho$ zeros at each row, and $\mathbf{0}_{n-2\rho}$ is a  $1 \times (n-2 \rho)$ matrix of zeros. 

\smallskip

Let $\calC$ be a row-LRPC code admitting $H$ as its parity check matrix.
We can write all the possible equivalent parity check matrices for $\calC$ as $H' = T H$ for some invertible $(n-k)\times(n-k)$ matrix $T$ over $\Fqm$.
Since $T$ is a square invertible matrix, all its columns have to be linearly independent.
Consider the first two columns of $T$, without loss of generality we can assume that the submatrix
\[T_{\{1,2\}, \{1,2\}} =
\begin{pmatrix}
t_{1,1} & t_{1,2}\\
t_{2,1} & t_{2,2}\\
\end{pmatrix}
\]
has non null determinant. The first $2$ rows of $TH$ can be given by
\[
\begin{pmatrix}
t_{1,1} a_1   & \cdots  &  t_{1,1} a_\rho     & t_{1,2} b_1  & \cdots   & t_{1,2} b_\rho & \bx(T)    \\
t_{2,1} a_1 & \cdots & t_{2,1} a_\rho & t_{2,2} b_1 & \cdots &t_{2,2} b_\rho & \by(T)
\end{pmatrix},
\]
where $\bx(T)$ and $\by(T)$ are two vectors of length $(n-2\rho)$ depending on $T$ (which are both $0$-vectors in the worst case).

Let $\calT$ be the $\Fq$-vector subspace generated by the elements of the submatrix
$TH_{\{1,2\}, \{1, \ldots, 2\rho\}}$ we can describe it as
\[
\calT = t_{1,1} \calA + t_{2,1} \calA + t_{1,2} \calB + t_{2,2} \calB.
\]
Since $T_{\{1,2\}, \{1,2\}}$ is invertible, it follows that either $t_{1,1} t_{2,2} \ne 0$ or $t_{1,2} t_{2,1} \ne 0$.
We want to prove that $\dim(\calT) > \rho$. Without loss of generality we can assume $t_{1,1}, t_{2,2} \ne 0$, since $t_{1,1} \calA + t_{2,2} \calB \subseteq \calT$ then
\[
\dim(t_{1,1}\calA + t_{2,2} \calB ) \le \dim(\calT).
\]
Observe that 
\begin{align*}
 &\dim(t_{1,1} \calA + t_{2,2} \calB ) \\
=  & \dim(t_{1,1} \calA) + \dim(t_{2,2} \calB )  - \dim(t_{1,1} \calA \cap t_{2,2} \calB ) \\
 = & 2\rho - \dim(t_{1,1} \calA \cap t_{2,2} \calB ). 
\end{align*}
This implies $\dim(\calT) \ge 2\rho - \dim(t_{1,1} \calA \cap t_{2,2} \calB )$.
Since $\dim(t_{1,1} \calA \cap t_{2,2} \calB ) \le \rho$, then
$\dim(\calT) \ge \rho$, where the equality holds only if
$\dim(t_{1,1} \calA \cap t_{2,2} \calB) = \rho$. This is actually equivalent to 
\[
\calA = t_{1,1}^{-1} t_{2,2}\calB,
\]
which is a contradiction since by our choice of $\calA$ and $\calB$ there exists no $\lambda \in \Fqm$ satisfying $\calA = \lambda \calB$.
Let $\mathrm{rw}(H)$ denote the rank weight of the matrix $H$. For any 
invertible matrix $T$, we have
    \begin{align*}
        \mathrm{rw}(TH) & \ge \mathrm{rw}(TH_{\{1,2\}, \{1,\ldots, n\}}) 
        \\
        &\ge \mathrm{rw}(TH_{\{1,2\}, \{1,\ldots, 2\rho\}}) > \rho.
    \end{align*}
The desired conclusion follows.
\end{proof}

\begin{remark}
Theorem 1 is does not cover the whole range of all parameters that make Theorem 1's result hold. Indeed it is quite easy to find such examples for  parameters. 
\end{remark}

Theorem 1 indicates that the extension of LRPC codes to row-LRPC codes is a proper extension for $\rho\ge 2$. 
Nevertheless, this is not true for the case when the row weight equals $1$.

\begin{lemma}\thlabel{weight1}
$ \small{\text{row-} \mathrm{LRPC}_q(n,m,k,1) =\mathrm{LRPC}_q(n,m,k,1).}$
\end{lemma}
\begin{proof}
Let $\mathcal{C}$ be a row-LRPC code of dimension $k$ and row weight $1$. The code $\calC$ admits an $(n-k)\times n$ parity check matrix $H$ whose rows have all rank weight $1$. This means $h_{i,j} \in \mathcal{H}_i = \langle h_i \rangle$ for some $h_i \in \Fqm^*$, i.e., $h_{i,j} = \lambda_{i,j} h_i$ for some $\lambda_{i,j} \in \Fq$.

The weight of the matrix $H$ is $\dim(\langle h_1, \ldots, h_{n-k}\rangle)$ which can be any number between $1$ and $\min(n-k, m)$. We want to prove that there exists as well another parity check matrix $H'$ of weight $1$ for the same code $\calC$. In fact, consider the diagonal matrix $L$ with the entries $(i,i)$ being $h_i^{-1}$ and all other entries zero.
Then the matrix $H'= LH$ is still a parity check matrix of $\calC$ and its entries can be written as $h_{i,j}' = h_i^{-1} h_{i,j} = h_i^{-1} \lambda_{i,j} h_i = \lambda_{i,j} \in \Fq$. Therefore the weight of $H'$ is $1$.
\end{proof}

\section{Decoding of Row-LRPC Codes}
The syndrome decoding problem for Hamming distance is a fundamental problem in complexity theory, which was
	proved to be an NP-complete problem Berlekamp, McEliece and van Tilborg \cite{RSD-1978-TIT}. When the code is embedded with rank metric, the corresponding rank syndrome decoding (RSD) problem is given as below.
	
	\begin{definition} \cite{GPT-1991,GaboritRuattaSchrek2016}
		Given an $(n-k)\times n$ matrix over $\Fqm$, a vector $\bs$ in $\Fqm^{n-k}$ and an integer $r$,
		find a vector $\be \in \Fqm^n$ of rank weight $\mathrm{rw}(\be) = r$ satisfying
		\(
		\be H^\intercal = \bs.
		\)
	\end{definition}
	
	By  embedding linear codes in the Hamming space into
	linear codes over an extension field equipped with the rank
	metric, Gaborit and Z\'{e}mor \cite{GaboritZemor2016} proved the hardness of
	the  RSD problem under an unfaithful random
	reduction, which indicates that if the RSD problem is in ZPP, then ZPP = NP.
 	 The RSD problem has been intensively studied in recent years driven by its application in post-quantum cryptosystems \cite{Gaborit2014-RankCrypto, Loidreau2017, ROLLO}.
	Early attacks on RSD appeared in \cite{ChabaudStern1996,OurivskiJohansson2002}.
	Some progress was lately made in \cite{GaboritRuattaSchrek2016}, where the authors proposed two new generic attacks: \textit{support attack} and \textit{polynomial annulator attack} for the RSD problem,
	which outperform the combinatorial attack in \cite{OurivskiJohansson2002} when the error vector has a relatively small rank weight $r$.
	Very recently, algebraic attacks with Grobner basis technique \cite{AlgAttack2020} and without Grobner basis technique \cite{AlgAttack2020-b} were proposed, which have significantly advanced the RSD problem.
	
	With the hardness of RSD problem, by far there are two general families of rank metric codes for which 
	(probabilistic) polynomial-time decoding algorithms have been found. Known decoding algorithms for the Gabidulin codes (e.g., \cite{Gabidulin1985})
	rely on the strong algebraic structure of Gabidulin codes, and the probabilistic decoding algorithm for the LRPC codes \cite{gaborit2013} utilizes 
	the low-weight property of the codes. More specifically,
	the decoding algorithm for LRPC codes can be divided into two steps. The first step is to recover the support $\calE$ of the error vector $\be$ by intersecting the vector spaces derived from the parity check matrix $H$ and the syndrome vector $\mathbf{s}$; the second step is to reconstruct each component of the error vector in the support $\calE$ by re-utilizing the syndrome equation $\mathbf{s}=\mathbf{e}H^\intercal$. 
	
	\smallskip
	
	When a row-LRPC code of row weight $\rho \geq 2$ is considered, its parity check matrix $H$ may have weight $\delta$ ranging from  $\rho$ to $\min(m, (n-k)\rho)$. Even in the non generic case when there exists a parity check matrix $H'$ for the same code of weight sufficiently small, there is no efficient algorithm to compute $H'$ from a randomly given $H$. Therefore, in general, we cannot use the decoding algorithm of LRPC codes for decoding row-LRPC codes.
	
	\smallskip
	
	In the following we shall explore the decoding process of row-LRPC codes, particularly when the row weight is sufficiently small.
	Similar to the two families of efficiently decodable rank metric codes Gabidulin codes \cite{Gabidulin1985} and LRPC codes \cite{gaborit2013,LRPC-2019-TIT}, the critical step of decoding row-LRPC codes is to efficiently recover the error support.
	Once the the error support is recovered, it is possible to expand the $n-k$ linear equation given by the syndrome  obtaining $m(n-k)$ equation over $\Fq$ and solve the linear system.

\subsection{The simplest case $r=1$}

The easiest case to treat is when the error support $\calE$ has dimension $1$. In this case, we can write $\calE = \langle \varepsilon \rangle$ for some $\varepsilon \in \Fqm$. The error $\be$ can be written as $(\lambda_1 \varepsilon, \ldots, \lambda_n \varepsilon)$ for some $\lambda_j \in \Fq$. The general syndrome equation is then given by
\begin{equation}\label{syndrome:error1}
s_i = \sum_{j=1}^n h_{i,j} \lambda_j \varepsilon = \varepsilon  \sum_{j=1}^n h_{i,j} \lambda_j = \varepsilon h_i, \qquad h_i \in \calH_i,
\end{equation}
where $\calH_i$ is the $\rho$-dimensional $\Fq$-vector subspace generated by the elements of the $i$-th row of $H$.
From equation \eqref{syndrome:error1}, if $s_i \ne 0$ then $h_i \ne 0$, so we have $\varepsilon \in s_i \calH_i^{-1}$. That is equivalent to  $\varepsilon^{-1} \in s_i^{-1} \calH$. This second formulation has an advantage that $s_i^{-1} \calH$ is an $\Fq$-vector subspace since
$$\lambda_1s_i^{-1}h_1+\lambda_2s_i^{-1}h_2 = s_i^{-1}(\lambda_1h_1+\lambda_2h_2)$$ belongs to $s_i^{-1} \calH$ for any $\lambda_1, \lambda_2 \in \Fq$, $h_1, h_2\in \calH$. By intersecting several sets of the form $s_i^{-1} \calH_i$, we will likely obtain the $\Fq$-vector space $\Fq \varepsilon^{-1}$ from which it's easy to find the support $\calE$.

\subsection{Cramer's rule decoding for $r=2$}
We will present a method to recover the error support when the error has rank $r \geq 2$. Before treating the general case, we shall start with the case $r=2$ as a simpler illustration in this subsection.

Notice that, as previously shown, the set $s \calH = \{sh| h \in \calH\}$ is a $\Fq$-vector subspace
if $\calH $ is a $\Fq$-vector subspace. Moreover, it is clear that $s \calH =  \langle s h_1, \ldots, s h_\rho \rangle $ when $\calH =  \langle h_1, \ldots,  h_\rho \rangle$ since any element in $s \calH$ can be linearly expressed by $sh_1, \ldots, sh_\rho$. This fact will be frequently used in the subsequent decoding process.

\smallskip

Let $\calE = \langle \varepsilon_1, \varepsilon_2 \rangle$ be the error support of $\be$ and let $H$ be a parity check matrix of row weight $\rho$. The general syndrome equation will be of the form
\[
h_{i,1} \varepsilon_1 + h_{i,2} \varepsilon_2 = s_i, \qquad h_{i,1}, h_{i,2} \in \calH_i 
\]
where $\calH_i$ is the $\rho$-dimensional $\Fq$-vector subspace generated by the elements of the $i$-th row of the parity check matrix $H$.
When $s_i \ne 0$ we can multiply the whole equation by $s_i ^{-1}$. Let $\calA_i$ be  an $\Fq$-vector space given by
\[
\calA_i = \begin{cases}
    s_i^{-1} \calH_i, & \text{if } s_i\neq 0, \\
    \calH_i, & \text{otherwise}.
\end{cases}
\] 
In this way, we can rewrite the equation $H\be^\intercal = \bs^\intercal$ as follows
 \begin{equation}\label{rank2_system_matrix}
    A \cdot \be^\intercal = \bv \Leftrightarrow
    \begin{pmatrix}
        a_{1,1}   & a_{1,2}    \\
        a_{2,1}   & a_{2,2}    \\
        \vdots    & \vdots     \\
        a_{n-k,1} &  a_{n-k,r} 
    \end{pmatrix}
    \begin{pmatrix}
        \varepsilon_1 \\
        \varepsilon_2 
    \end{pmatrix}
    =\begin{pmatrix}
        v_1 \\
        v_2 \\
        \vdots
        \\
        v_{n-k}
    \end{pmatrix},
\end{equation}
where $a_{i,j} \in \calA_i$ and $\bv=(v_1,\ldots, v_{n-k})^\intercal$ is a vector over $\Fqm$ such that $v_i = 0$ when $s_i=0$ and is $1$ otherwise.

\smallskip

Suppose that $v_i= v_j = 1$, from \eqref{rank2_system_matrix} we can extract the linear system
 \begin{equation}\label{rank2_reduced_system_matrix}
    A_{\{i,j\}, \{1,2\}} \cdot \be^\intercal =
    \begin{pmatrix}
        a_{i,1}   & a_{i,2}    \\
        a_{j,1}   & a_{j,2}    \\
    \end{pmatrix}
    \begin{pmatrix}
        \varepsilon_1 \\
        \varepsilon_2 
    \end{pmatrix}
    =
        \begin{pmatrix}
        1 \\
        1
    \end{pmatrix}.
\end{equation}
There are at most $q^{4\rho}$ possible configuration of the matrix $A_{\{i,j\}, \{1,2\}}$. 
If the matrix $A_{\{i,j\}, \{1,2\}}$ is invertible, we can solve \eqref{rank2_reduced_system_matrix} by Cramer's rule \cite{KramelRule} as
\begin{equation}
    \varepsilon_1 = \frac{ \det 
    \begin{pmatrix}
        1   & a_{i,2} \\
        1   & a_{j,2}
    \end{pmatrix}
    }
    {\det
    \begin{pmatrix}
        a_{i,1}   & a_{i,2} \\
        a_{j,1}   & a_{j,2}
    \end{pmatrix}
    }, \quad 
    \varepsilon_2 = \frac{ \det 
    \begin{pmatrix}
        a_{i,1}   & 1 \\
        a_{j,1}   & 1
    \end{pmatrix}
    }
    {\det
    \begin{pmatrix}
        a_{i,1}   & a_{i,2} \\
        a_{j,1}   & a_{j,2}
    \end{pmatrix}
    }.
\end{equation}
After calculation of determinants this reduces to
\begin{equation}\label{error_described}
    \varepsilon_1 = \frac{
        a_{i,2} - a_{j,2}
    }
    {
        a_{i,1} a_{j,2} - a_{i,2} a_{j,1}
    }, \quad 
    \varepsilon_2 = \frac{
        a_{i,1} - a_{j,1}
    }
    {
        a_{i,1} a_{j,2} - a_{i,2} a_{j,1}
    }.
\end{equation}
We define the set
\begin{equation}\label{Bij_rank2}
    B_{i,j} = \left\{
    \frac{
        a_{i,2} - a_{j,2}
    }
    {
        a_{i,1} a_{j,2} - a_{i,2} a_{j,1}
    }
    \,:\, 
    a_{i,*} \in \calA_i, \,a_{j,*} \in \calA_j
    \right\}
    .
\end{equation}
Recall that $\calA.\calB$ denotes the product vector space generated by the products of elements in the vector spaces $\calA$ and $\calB$.
For the set $B_{i,j}$ in \eqref{Bij_rank2}, the denominators in the division 
satisfy
$$
a_{i,1}a_{j,2} - a_{i,2}a_{j,1} \in \calA_i . \calA_j
$$ 
since $a_{i,*}a_{j,*}\in \calA_i . \calA_j$. This leads to an observation that
$$
B_{i,j} \subseteq \frac{\calA_i + \calA_j}{\calA_i . \calA_j}.
$$ 
In addition, due to the symmetry of the equations in \eqref{error_described}, we have 
$$
\calE \subsetneq B_{i,j}\subseteq \frac{\calA_i + \calA_j}{\calA_i . \calA_j}.
$$ Since this is true for any $B_{i,j}$ we can intersect different $B_{i,j}$ in order to get a smaller set where to look for the error support $\calE$.
Since the sets $B_{i,j}$ depend on $4$ parameters in $\Fq$-vector spaces of dimension $\rho$, their size is upper bounded by $q^{4\rho}$. 

Following the method just explained, we run some MAGMA simulations for the case $r=2$, $\rho=2$, $n=m=20$ and $q=2,3$. Using $4$ rows in the parity check matrix it was possible to recover the error support in $707/1000$ cases for $q=2$ and in $954/1000$ cases for $q=3$. Increasing the number of parity check rows to $6$, the results improved to $988/1000$ and $998/1000$, respectively.
In all of the above mentioned cases, it was possible to recover the error support after the intersection of just two $B_{i,j}$ sets.

In Section \ref{SecV} we will explain why and how increasing the number of rows of the parity check matrix can help us achieve higher success rates in decoding.

\subsection{Cramer's rule for the general case}
The decoding process with Cramer's rule for $r=2$ can be easily extended to general cases.

\smallskip

Let $\calE = \langle \varepsilon_1, \ldots, \varepsilon_r \rangle$ be the support of the error $\be=(e_1,\ldots, e_n)$. Given $H$ a parity check matrix of row weight $\rho$ we can rewrite the $i$-th syndrome equation $\be\cdot \bh_i^\intercal = s_i$ as
\begin{equation}\label{syndorme_eq}
    \varepsilon_1 h_{i,1} + \cdots + \varepsilon_r h_{i,r} = s_i \qquad h_{i,*} \in \calH_i,  
\end{equation}
where $\calH_i = \langle\bh_i\rangle$ is the $\rho$-dimensional vector space generated by the $i$-th row of $H$.

As we did in the previous section, we can manipulate each equations multiplying on both sides by $s_i^{-1}$ when $s_i \ne 0$. We rename the variables and the $\Fq$-vector subspace with the same notation.
In this way we get a system that can be described by the matrix equation $A \cdot \be^\intercal =
    \bv $, more specifically,
 \begin{equation}\label{general_system_matrix}
    \begin{pmatrix}
        a_{1,1}   & \cdots & a_{1,r}    \\
        a_{2,1}   & \cdots & a_{2,r}    \\
        \vdots    &  \ddots      & \vdots     \\
        a_{n-k,1} & \cdots & a_{n-k,r} 
    \end{pmatrix}
    \begin{pmatrix}
        \varepsilon_1 \\
        \vdots \\
        \varepsilon_r 
    \end{pmatrix}
    =\begin{pmatrix}
        v_1 \\
        v_2 \\
        \vdots
        \\
        v_{n-k}
    \end{pmatrix}.
\end{equation}
This time we can select $r$ rows of \eqref{general_system_matrix} and solve by the Cramer's rule.
Let $I = \{i_1, \ldots, i_r\} \subseteq \{1, \ldots, n-k \}$ a set of $r$ indices. The Cramer's rule for the sub-system given by the the rows indexed by $I$ is 
\begin{equation}\label{general_cramer}
 \varepsilon_j = \frac{\det( A_{I, \{1, \ldots, j-1\}} | \bv | A_{I, \{j+1, \ldots, r\}})}{\det(A_{I, \{1, \ldots, r\}})},
\end{equation} 
 where the matrix $A_{I, \{1, \ldots j-1\}} | \bv | A_{I, \{j+1, \ldots, r\}}$ is the matrix $A_{I, \{1, \ldots, r\}}$ where the $j$-th column is substituted by $\bv$.
In this way we can define the set $B_{I} \subseteq \Fqm$ of size upper bounded by $q^{r^2\rho}$ as:
\begin{equation}\label{general_B_cramer}
 B_I = \left\{ \frac{\det( A_{I, \{1, \ldots, r-1\}} | \bv )}{\det(A_{I, \{1, \ldots, r\}})} \bigg|\quad a_{i,j} \in \calA_i \right\}
\end{equation}
where $a_{i,j}$ indicates the entries of $A_{I, \{1, \ldots, r\}}$ at the row $i$ and column $j$. Notice that since the variables $a_{i,j}$ are uniformly chosen from all the possible values of the $\Fq$-vector space $\calA_i$, the quantity $\det(A^{'}_{I, \{1, \ldots j-1\}} | \bv | A^{'}_{I, \{j+1, \ldots, r\}})$ with  $j=1, \ldots, r$ can be obtained as well as $\det( A_{I, \{1, \ldots, r-1\}} | \bv )$, where $A$ is obtained from $A'$ by renaming the columns after $j$ and multiplying one column by $-1$ when $\Fq$ has odd characteristic. 

Each $B_I$ contains the error support $\calE$. We can then intersect some different $B_I$'s to obtain a smaller set that can be exactly $\calE$ or at least a smaller set where we know $\calE$ is contained. The situation where we expect this algorithm to work is when the cardinality of $B_I$ is much smaller than $\Fqm$ (i.e. $r^2\rho \ll m$).

If we intersect two sets $B_I$ and $B_J$ such that $I \cap J \ne \emptyset$, we can expect to have a bigger intersection.
Suppose without loss of generality that the index $1$ belongs to both $I$ and $J$. Among all the possible configurations of $A_{I, \{1, \ldots, r\}}$ and $A_{J, \{1, \ldots, r\}}$ we will have one where the first row of both matrices is $(0, \ldots, 0, a_{1,r})$. When we apply the Cramer's rule to that matrix we will obtain $a_{1,r}^{-1}$, then we know without the need of computing that $\calA_1^{-1} \subset B_I \cap B_J$.

For this reason, in order to do as few intersections as possible, it is better to use disjointed indexes. If we want to have at least two disjointed sets of indexes, we need $n-k \ge 2r$. Suppose we managed to recover the error support, to recover each component of the error $\be$ we need to solve an expanded linear system of $nr$ unknowns and $m(n-k)$ equations. Therefore we need that $(n-k) \ge \lceil \frac{nr}{m} \rceil$. Then we have the condition on the dimension $k$
\[
k \le \min \left(n-2r, n- \left \lceil \frac{nr}{m} \right \rceil \right).
\]
The hardest part of this algorithm is ruled by the error support recovery. In particular it is given by the cost of computing the sets $B_I$, to obtain each element of $B_I$ we need to solve a linear system in $r$ unknowns. We can then give an estimation of the total complexity in the order of $O(r^3 q^{r^2 \rho})$.

\section{Decoding failures}\label{SecV}
We implement the algorithm described in the previous section on MAGMA on an easy instance of the problem. In particular we created different instances of row-LRPC codes of row weight $2$ and we did correct errors of rank weight $2$.
During the error support recovery, we noticed two possible causes of failure.

\subsection{Syndrome contains some zeroes}
Let $\calE = \langle \varepsilon_1, \ldots, \varepsilon_r \rangle$ be the error support. The generic equation of the system is
\begin{equation}\label{general_equation}
a_{i,1} \varepsilon_1 + \cdots + a_{i,r} \varepsilon_r = s_i.
\end{equation}
 
We can consider $a_{i,*}$ as independent variables uniformly distributed over $\calH_i$, then $a_{i,j} \varepsilon_j$ are independent uniformly distributed over $\calH_i \varepsilon_j$. Therefore $s_i$ is a variable uniformly distributed over the product space $\calH_i . \calE$ which, with an high probability, has dimension $r \rho$ when $r \rho \ll m$ \cite{argon2018}.\\
By our hypothesis 
$\calH_i = \langle \alpha_1, \ldots, 
\alpha_\rho \rangle$ then each $a_{i,j}$ can be rewritten as 
$\sum_{i \in \{1, \ldots, \rho \}} \lambda_{i,j} \alpha_i$ for some 
$\lambda_{i,j} \in \Fq$. The equation \eqref{general_equation} can be rewritten as
\[
\sum_{j \in \{1,\ldots, r\}} \sum_{l \in \{1, \ldots, \rho\}} \lambda_{l,j} \alpha_{l} \varepsilon_j = s_i \quad \lambda_{l,j} \in \Fq.
\]
Suppose now that $s_i=0$ and $r \rho \ll m$, with an high probability, $\alpha_l \varepsilon_j$ are all $\Fq$-linearly independent which implies $\lambda_{i,j} = 0 \quad \forall i \in \{1, \ldots, \rho\}, \quad  \forall j \in \{1, \ldots, r\}$, therefore $a_{i,1} = \cdots = a_{i,\rho} = 0$.\\
Unfortunately in the case $s_i = 0$ we cannot get any information about the error support. Indeed, if all the coefficients $a_{i,*}$ are zero, the equation \eqref{general_equation} is trivially satisfied by all the possible sets of $r$ elements in $\Fqm$.

If we consider $s_i$ as independent uniformly distributed variables over the $\Fq$-vector subspace  $\calH_i . \calE$ of dimension $r \rho$, we have the probability  $P(s_i = 0) = q^{-r \rho}$. We expect to have at least $t$ non-zero components of the syndrome with probability 
$\sum_{i=t}^{n-k} \binom{n-k}{i}(1-q^{-r\rho})^i q^{-r\rho(n-k-i)}$.

The equations leading to zero will not be totally discarded as they can still be used, after the recovery of the error support, to help us recovering the actual error component by component.

\subsection{Intersecting subspaces}
For simplicity here we will analyze the case when the error support has rank $2$.
In the error support recovery, after we discard all the equation corresponding to $s_i = 0$, there is still a possible issue we will show in this paragraph.

Let $\calE = \langle \varepsilon_1,\varepsilon_2 \rangle$
be the error support, consider the generic system:
\begin{equation}\label{2_rows_intersect}
    \begin{cases} 
        a_{i,1} \varepsilon_1 + a_{i,2} \varepsilon_r = 1 \\
        a_{j,1} \varepsilon_1 + a_{j,2} \varepsilon_r = 1 
    \end{cases}
\end{equation}
where $a_{i,*} \in \calA_i = \calH_i s_i^{-1}$ and $a_{j,*} \in \calA_j = \calH_j s_j^{-1}$.

Suppose that $\calA_i$ and $\calA_j$ intersect non-trivially, the associated set $B_{i,j}$ will be smaller than what we usually expect. A smaller set is not necessarily a good news. At the contrary, what we observed with MAGMA, is that many times this set does contain only part of the error support. When we intersect such a set with the other sets containing the whole support, we recover only part of the support instead of the whole as we wish.

Given two randomly generated $\Fq$-vector subspaces of dimension $\rho \ll m$, the probability of having a non-trivial intersection between the two should be negligible. What we observe with MAGMA, suggests that in our case this happens with a much higher frequency than one could naively expect.

Consider better the system \eqref{2_rows_intersect}, without loss of generality assume $i=1$ and $j=2$
\[
\begin{cases} 
    a_{1,1} \varepsilon_1 + a_{1,2} \varepsilon_2 = 1 \\
    a_{2,1} \varepsilon_1 + a_{2,2} \varepsilon_2 = 1 
\end{cases}
\]
Consider the case $a_{1,2} = a_{2,2} = 0$, this would imply $a_{1,1} \varepsilon_1 = a_{2,1} \varepsilon_1= 1$, then $\varepsilon_1^{-1} \in \calA_1 \cap \calA_2$. A similar conclusion can be made from $a_{1,1} = a_{2,1} = 0$. A bit less straight forward is the case $a_{1,1} = a_{1,2} \ne 0$. In this case we have $(\varepsilon_1 + \varepsilon_2)^{-1} \in \calA_1$, if we also have $a_{2,1} = a_{2,2} \ne 0$ then $(\varepsilon_1 + \varepsilon_2)^{-1} \in \calA_1 \cap \calA_2$.
The coefficients $a_{i,*}$ all belong to $\calA_i$ which, by our hypothesis, is a small set of size $q^\rho$. This explains why we observe non trivial intersections so often.

In the general case, after we consider a system of $r$ rows, we will obtain an $r \times r$ matrix $A$ whose entries belong to some $\Fq$-vector subspace of dimension $\rho$, that multiplies the vector $\be = (\varepsilon_1, \ldots, \varepsilon_r)$ containing a base of the error support. Intuitively, in the case the $j$-th column of $A$ is all $0$s,  we have no hope to get some information about $\varepsilon_j$. The likelihood of this event decreases exponentially in $r$. In the case $r > 2$, we cannot find a relation between the $\Fq$-vector subspaces $\calA_i$ that tells us something about the matrix $A$.  
Increasing the number of equation will increase the probability of success.
In the case of intersecting spaces, it will give us more $\Fq$-vector subspaces among which we can find more disjointed couples. In the case of too many zeroes in the syndrome, increasing the number of equations will increase the probability of having enough non-zero components.

\section{Conclusion}

This paper introduces a new family of rank metric codes, which can be seen as an extended variant of LRPC codes introduced in \cite{gaborit2013}. We showed that the proposed row-LRPC codes contain LRPC codes as a strictly proper subfamily when row weight is no less than $2$ under certain condition, which confirms the validity of the extension. The decoding of row-LRPC codes is also studied by intersecting certain sets obtained from the syndrome equation, which has complexity upper bounded by $O(r^3q^{r^2\rho})$, where $\rho$ is the row weight of the code in question and $r$ is the rank weight of the error vector.

Due to high complexity of the current decoding approach, we will continue exploring the possibility of a polynomial-time decoding algorithm for the row-LRPC codes, which by then would have interesting applications in cryptography..

\bibliographystyle{IEEEtran}

\bibliography{bibliography.bib}
\end{document}